\documentclass[a4paper,noarxiv,twocolumn]{quantumarticle}
\pdfoutput=1

\usepackage[utf8]{inputenc}
\usepackage[T1]{fontenc}
\usepackage[british]{babel}

\usepackage{amsmath}
\usepackage{amssymb}
\usepackage{amsthm}
\usepackage{amsfonts}
\usepackage{bm}
\usepackage{dsfont}
\usepackage{mathtools}
\usepackage{braket}
\usepackage{physics}

\usepackage{graphicx}
\usepackage[table]{xcolor}
    \definecolor{myblue}{RGB}{3,70,143}
    \definecolor{mypink}{RGB}{255,151,151}
    \definecolor{mypurple}{RGB}{164,0,102}
    \definecolor{myteal}{RGB}{0,157,158}
    \definecolor{darkred}{rgb}{0.5,0,0}
\usepackage{graphicx}
\usepackage{subfigure}
\usepackage{enumerate}
\usepackage[breakable]{tcolorbox}
\usepackage{booktabs}
\usepackage[normalem]{ulem}

\definecolor{pgreen}{RGB}{84, 129, 102}
\definecolor{porange}{RGB}{199, 103, 42}
\usepackage[
colorlinks = true,
citecolor=pgreen,
linkcolor=porange,
urlcolor=pgreen
]{hyperref}
\usepackage[capitalize,nameinlink]{cleveref}

\newtheorem{theorem}{Theorem}
\newtheorem{result}{Result}

\theoremstyle{definition}
\newtheorem{definition}{Definition}

\theoremstyle{remark}
\newtheorem{remark}{Remark}

\renewcommand*\paragraph[1]{\emph{#1.}---}
\newcommand*\eye{{\normalfont\openone}}

\renewcommand*\tilde[1]{\widetilde{#1}}

\renewcommand{\leq}{\leqslant}

\renewcommand{\geq}{\geqslant}
\renewcommand{\set}[1]{\{ #1 \}}

\newcommand{\siegen}{Naturwissenschaftlich-Technische Fakult\"{a}t, Universit\"{a}t Siegen, Walter-Flex-Stra\ss e 3, 57068 Siegen, Germany}

\newcommand{\unicamp}{Instituto de Matem\'{a}tica, Estat\'{i}stica e Computa\c{c}\~{a}o Cient\'{i}fica, Universidade Estadual de Campinas, 13083-859, Campinas, S\~{a}o Paulo, Brazil}

\newcommand{\sorbonne}{Sorbonne Universit\'{e}, CNRS, LIP6, F-75005 Paris, France}

\newcommand{\perimeter}{Perimeter Institute for Theoretical Physics, Waterloo, Ontario, N2L 2Y5, Canada}

\newcommand{\waterloo}{Institute for Quantum Computing, University of Waterloo, Waterloo, Ontario, N2L 3G1, Canada}

\newcommand{\inria}{CPHT, LIX, CNRS, Inria, École polytechnique, Institut Polytechnique de Paris, Palaiseau, France}

\newcommand{\siegenfinancialsupport}{the Deutsche Forschungsgemeinschaft (DFG, German Research Foundation, project numbers 447948357 and 440958198), the Sino-German Center for Research Promotion (Project M-0294), the ERC (Consolidator Grant 683107/TempoQ), the German Ministry of Education and Research (Project QuKuK, BMBF Grant No.\ 16KIS1618K), the House of Young Talents of the University of Siegen, }

\usepackage[breakable]{tcolorbox}
\usepackage{fancyvrb} 

\definecolor{incolor}{HTML}{303F9F}
\definecolor{outcolor}{HTML}{D84315}
\definecolor{cellborder}{HTML}{CFCFCF}
\definecolor{cellbackground}{HTML}{F7F7F7}

\makeatletter
\newcommand{\boxspacing}{\kern\kvtcb@left@rule\kern\kvtcb@boxsep}
\makeatother
\newcommand{\prompt}[4]{
    {\ttfamily\llap{{\color{#2}[#3]:\hspace{3pt}#4}}\vspace{-\baselineskip}}
}

\DefineVerbatimEnvironment{Highlighting}{Verbatim}{commandchars=\\\{\}}
\makeatletter
\def\PY@reset{\let\PY@it=\relax \let\PY@bf=\relax%
    \let\PY@ul=\relax \let\PY@tc=\relax%
    \let\PY@bc=\relax \let\PY@ff=\relax}
\def\PY@tok#1{\csname PY@tok@#1\endcsname}
\def\PY@toks#1+{\ifx\relax#1\empty\else%
    \PY@tok{#1}\expandafter\PY@toks\fi}
\def\PY@do#1{\PY@bc{\PY@tc{\PY@ul{%
    \PY@it{\PY@bf{\PY@ff{#1}}}}}}}
\def\PY#1#2{\PY@reset\PY@toks#1+\relax+\PY@do{#2}}

\@namedef{PY@tok@w}{\def\PY@tc##1{\textcolor[rgb]{0.73,0.73,0.73}{##1}}}
\@namedef{PY@tok@c}{\let\PY@it=\textit\def\PY@tc##1{\textcolor[rgb]{0.24,0.48,0.48}{##1}}}
\@namedef{PY@tok@cp}{\def\PY@tc##1{\textcolor[rgb]{0.61,0.40,0.00}{##1}}}
\@namedef{PY@tok@k}{\let\PY@bf=\textbf\def\PY@tc##1{\textcolor[rgb]{0.00,0.50,0.00}{##1}}}
\@namedef{PY@tok@kp}{\def\PY@tc##1{\textcolor[rgb]{0.00,0.50,0.00}{##1}}}
\@namedef{PY@tok@kt}{\def\PY@tc##1{\textcolor[rgb]{0.69,0.00,0.25}{##1}}}
\@namedef{PY@tok@o}{\def\PY@tc##1{\textcolor[rgb]{0.40,0.40,0.40}{##1}}}
\@namedef{PY@tok@ow}{\let\PY@bf=\textbf\def\PY@tc##1{\textcolor[rgb]{0.67,0.13,1.00}{##1}}}
\@namedef{PY@tok@nb}{\def\PY@tc##1{\textcolor[rgb]{0.00,0.50,0.00}{##1}}}
\@namedef{PY@tok@nf}{\def\PY@tc##1{\textcolor[rgb]{0.00,0.00,1.00}{##1}}}
\@namedef{PY@tok@nc}{\let\PY@bf=\textbf\def\PY@tc##1{\textcolor[rgb]{0.00,0.00,1.00}{##1}}}
\@namedef{PY@tok@nn}{\let\PY@bf=\textbf\def\PY@tc##1{\textcolor[rgb]{0.00,0.00,1.00}{##1}}}
\@namedef{PY@tok@ne}{\let\PY@bf=\textbf\def\PY@tc##1{\textcolor[rgb]{0.80,0.25,0.22}{##1}}}
\@namedef{PY@tok@nv}{\def\PY@tc##1{\textcolor[rgb]{0.10,0.09,0.49}{##1}}}
\@namedef{PY@tok@no}{\def\PY@tc##1{\textcolor[rgb]{0.53,0.00,0.00}{##1}}}
\@namedef{PY@tok@nl}{\def\PY@tc##1{\textcolor[rgb]{0.46,0.46,0.00}{##1}}}
\@namedef{PY@tok@ni}{\let\PY@bf=\textbf\def\PY@tc##1{\textcolor[rgb]{0.44,0.44,0.44}{##1}}}
\@namedef{PY@tok@na}{\def\PY@tc##1{\textcolor[rgb]{0.41,0.47,0.13}{##1}}}
\@namedef{PY@tok@nt}{\let\PY@bf=\textbf\def\PY@tc##1{\textcolor[rgb]{0.00,0.50,0.00}{##1}}}
\@namedef{PY@tok@nd}{\def\PY@tc##1{\textcolor[rgb]{0.67,0.13,1.00}{##1}}}
\@namedef{PY@tok@s}{\def\PY@tc##1{\textcolor[rgb]{0.73,0.13,0.13}{##1}}}
\@namedef{PY@tok@sd}{\let\PY@it=\textit\def\PY@tc##1{\textcolor[rgb]{0.73,0.13,0.13}{##1}}}
\@namedef{PY@tok@si}{\let\PY@bf=\textbf\def\PY@tc##1{\textcolor[rgb]{0.64,0.35,0.47}{##1}}}
\@namedef{PY@tok@se}{\let\PY@bf=\textbf\def\PY@tc##1{\textcolor[rgb]{0.67,0.36,0.12}{##1}}}
\@namedef{PY@tok@sr}{\def\PY@tc##1{\textcolor[rgb]{0.64,0.35,0.47}{##1}}}
\@namedef{PY@tok@ss}{\def\PY@tc##1{\textcolor[rgb]{0.10,0.09,0.49}{##1}}}
\@namedef{PY@tok@sx}{\def\PY@tc##1{\textcolor[rgb]{0.00,0.50,0.00}{##1}}}
\@namedef{PY@tok@m}{\def\PY@tc##1{\textcolor[rgb]{0.40,0.40,0.40}{##1}}}
\@namedef{PY@tok@gh}{\let\PY@bf=\textbf\def\PY@tc##1{\textcolor[rgb]{0.00,0.00,0.50}{##1}}}
\@namedef{PY@tok@gu}{\let\PY@bf=\textbf\def\PY@tc##1{\textcolor[rgb]{0.50,0.00,0.50}{##1}}}
\@namedef{PY@tok@gd}{\def\PY@tc##1{\textcolor[rgb]{0.63,0.00,0.00}{##1}}}
\@namedef{PY@tok@gi}{\def\PY@tc##1{\textcolor[rgb]{0.00,0.52,0.00}{##1}}}
\@namedef{PY@tok@gr}{\def\PY@tc##1{\textcolor[rgb]{0.89,0.00,0.00}{##1}}}
\@namedef{PY@tok@ge}{\let\PY@it=\textit}
\@namedef{PY@tok@gs}{\let\PY@bf=\textbf}
\@namedef{PY@tok@gp}{\let\PY@bf=\textbf\def\PY@tc##1{\textcolor[rgb]{0.00,0.00,0.50}{##1}}}
\@namedef{PY@tok@go}{\def\PY@tc##1{\textcolor[rgb]{0.44,0.44,0.44}{##1}}}
\@namedef{PY@tok@gt}{\def\PY@tc##1{\textcolor[rgb]{0.00,0.27,0.87}{##1}}}
\@namedef{PY@tok@err}{\def\PY@bc##1{{\setlength{\fboxsep}{\string -\fboxrule}\fcolorbox[rgb]{1.00,0.00,0.00}{1,1,1}{\strut ##1}}}}
\@namedef{PY@tok@kc}{\let\PY@bf=\textbf\def\PY@tc##1{\textcolor[rgb]{0.00,0.50,0.00}{##1}}}
\@namedef{PY@tok@kd}{\let\PY@bf=\textbf\def\PY@tc##1{\textcolor[rgb]{0.00,0.50,0.00}{##1}}}
\@namedef{PY@tok@kn}{\let\PY@bf=\textbf\def\PY@tc##1{\textcolor[rgb]{0.00,0.50,0.00}{##1}}}
\@namedef{PY@tok@kr}{\let\PY@bf=\textbf\def\PY@tc##1{\textcolor[rgb]{0.00,0.50,0.00}{##1}}}
\@namedef{PY@tok@bp}{\def\PY@tc##1{\textcolor[rgb]{0.00,0.50,0.00}{##1}}}
\@namedef{PY@tok@fm}{\def\PY@tc##1{\textcolor[rgb]{0.00,0.00,1.00}{##1}}}
\@namedef{PY@tok@vc}{\def\PY@tc##1{\textcolor[rgb]{0.10,0.09,0.49}{##1}}}
\@namedef{PY@tok@vg}{\def\PY@tc##1{\textcolor[rgb]{0.10,0.09,0.49}{##1}}}
\@namedef{PY@tok@vi}{\def\PY@tc##1{\textcolor[rgb]{0.10,0.09,0.49}{##1}}}
\@namedef{PY@tok@vm}{\def\PY@tc##1{\textcolor[rgb]{0.10,0.09,0.49}{##1}}}
\@namedef{PY@tok@sa}{\def\PY@tc##1{\textcolor[rgb]{0.73,0.13,0.13}{##1}}}
\@namedef{PY@tok@sb}{\def\PY@tc##1{\textcolor[rgb]{0.73,0.13,0.13}{##1}}}
\@namedef{PY@tok@sc}{\def\PY@tc##1{\textcolor[rgb]{0.73,0.13,0.13}{##1}}}
\@namedef{PY@tok@dl}{\def\PY@tc##1{\textcolor[rgb]{0.73,0.13,0.13}{##1}}}
\@namedef{PY@tok@s2}{\def\PY@tc##1{\textcolor[rgb]{0.73,0.13,0.13}{##1}}}
\@namedef{PY@tok@sh}{\def\PY@tc##1{\textcolor[rgb]{0.73,0.13,0.13}{##1}}}
\@namedef{PY@tok@s1}{\def\PY@tc##1{\textcolor[rgb]{0.73,0.13,0.13}{##1}}}
\@namedef{PY@tok@mb}{\def\PY@tc##1{\textcolor[rgb]{0.40,0.40,0.40}{##1}}}
\@namedef{PY@tok@mf}{\def\PY@tc##1{\textcolor[rgb]{0.40,0.40,0.40}{##1}}}
\@namedef{PY@tok@mh}{\def\PY@tc##1{\textcolor[rgb]{0.40,0.40,0.40}{##1}}}
\@namedef{PY@tok@mi}{\def\PY@tc##1{\textcolor[rgb]{0.40,0.40,0.40}{##1}}}
\@namedef{PY@tok@il}{\def\PY@tc##1{\textcolor[rgb]{0.40,0.40,0.40}{##1}}}
\@namedef{PY@tok@mo}{\def\PY@tc##1{\textcolor[rgb]{0.40,0.40,0.40}{##1}}}
\@namedef{PY@tok@ch}{\let\PY@it=\textit\def\PY@tc##1{\textcolor[rgb]{0.24,0.48,0.48}{##1}}}
\@namedef{PY@tok@cm}{\let\PY@it=\textit\def\PY@tc##1{\textcolor[rgb]{0.24,0.48,0.48}{##1}}}
\@namedef{PY@tok@cpf}{\let\PY@it=\textit\def\PY@tc##1{\textcolor[rgb]{0.24,0.48,0.48}{##1}}}
\@namedef{PY@tok@c1}{\let\PY@it=\textit\def\PY@tc##1{\textcolor[rgb]{0.24,0.48,0.48}{##1}}}
\@namedef{PY@tok@cs}{\let\PY@it=\textit\def\PY@tc##1{\textcolor[rgb]{0.24,0.48,0.48}{##1}}}

\def\PYZus{\char`\_}
\def\PYZob{\char`\{}
\def\PYZcb{\char`\}}
\def\PYZca{\char`\^}
\def\PYZam{\char`\&}

\def\PYZgt{\char`\>}
\def\PYZsh{\char`\#}

\def\PYZhy{\char`\-}
\def\PYZsq{\char`\'}
\def\PYZdq{\char`\"}

\makeatother

\makeatletter
    \newbox\Wrappedcontinuationbox
    \newbox\Wrappedvisiblespacebox
    \newcommand*\Wrappedvisiblespace {\textcolor{red}{\textvisiblespace}}
    \newcommand*\Wrappedcontinuationsymbol {\textcolor{red}{\llap{\tiny$\m@th\hookrightarrow$}}}
    \newcommand*\Wrappedcontinuationindent {3ex }
    \newcommand*\Wrappedafterbreak {\kern\Wrappedcontinuationindent\copy\Wrappedcontinuationbox}
    \newcommand*\Wrappedbreaksatspecials {%
        \def\PYGZus{\discretionary{\char`\_}{\Wrappedafterbreak}{\char`\_}}%
        \def\PYGZob{\discretionary{}{\Wrappedafterbreak\char`\{}{\char`\{}}%
        \def\PYGZcb{\discretionary{\char`\}}{\Wrappedafterbreak}{\char`\}}}%
        \def\PYGZca{\discretionary{\char`\^}{\Wrappedafterbreak}{\char`\^}}%
        \def\PYGZam{\discretionary{\char`\&}{\Wrappedafterbreak}{\char`\&}}%
        \def\PYGZlt{\discretionary{}{\Wrappedafterbreak\char`\<}{\char`\<}}%
        \def\PYGZgt{\discretionary{\char`\>}{\Wrappedafterbreak}{\char`\>}}%
        \def\PYGZsh{\discretionary{}{\Wrappedafterbreak\char`\#}{\char`\#}}%
        \def\PYGZpc{\discretionary{}{\Wrappedafterbreak\char`\%}{\char`\%}}%
        \def\PYGZdl{\discretionary{}{\Wrappedafterbreak\char`\$}{\char`\$}}%
        \def\PYGZhy{\discretionary{\char`\-}{\Wrappedafterbreak}{\char`\-}}%
        \def\PYGZsq{\discretionary{}{\Wrappedafterbreak\textquotesingle}{\textquotesingle}}%
        \def\PYGZdq{\discretionary{}{\Wrappedafterbreak\char`\"}{\char`\"}}%
        \def\PYGZti{\discretionary{\char`\~}{\Wrappedafterbreak}{\char`\~}}%
    }
    \newcommand*\Wrappedbreaksatpunct {%
        \lccode`\~`\.\lowercase{\def~}{\discretionary{\hbox{\char`\.}}{\Wrappedafterbreak}{\hbox{\char`\.}}}%
        \lccode`\~`\,\lowercase{\def~}{\discretionary{\hbox{\char`\,}}{\Wrappedafterbreak}{\hbox{\char`\,}}}%
        \lccode`\~`\;\lowercase{\def~}{\discretionary{\hbox{\char`\;}}{\Wrappedafterbreak}{\hbox{\char`\;}}}%
        \lccode`\~`\:\lowercase{\def~}{\discretionary{\hbox{\char`\:}}{\Wrappedafterbreak}{\hbox{\char`\:}}}%
        \lccode`\~`\?\lowercase{\def~}{\discretionary{\hbox{\char`\?}}{\Wrappedafterbreak}{\hbox{\char`\?}}}%
        \lccode`\~`\!\lowercase{\def~}{\discretionary{\hbox{\char`\!}}{\Wrappedafterbreak}{\hbox{\char`\!}}}%
        \lccode`\~`\/\lowercase{\def~}{\discretionary{\hbox{\char`\/}}{\Wrappedafterbreak}{\hbox{\char`\/}}}%
        \catcode`\.\active
        \catcode`\,\active
        \catcode`\;\active
        \catcode`\:\active
        \catcode`\?\active
        \catcode`\!\active
        \catcode`\/\active
        \lccode`\~`\~
    }
\makeatother

\let\OriginalVerbatim=\Verbatim
\makeatletter
\renewcommand{\Verbatim}[1][1]{%
    \sbox\Wrappedcontinuationbox {\Wrappedcontinuationsymbol}%
    \sbox\Wrappedvisiblespacebox {\FV@SetupFont\Wrappedvisiblespace}%
    \def\FancyVerbFormatLine ##1{\hsize\linewidth
        \vtop{\raggedright\hyphenpenalty\z@\exhyphenpenalty\z@
            \doublehyphendemerits\z@\finalhyphendemerits\z@
            \strut ##1\strut}%
    }%
    \def\FV@Space {%
        \nobreak\hskip\z@ plus\fontdimen3\font minus\fontdimen4\font
        \discretionary{\copy\Wrappedvisiblespacebox}{\Wrappedafterbreak}
        {\kern\fontdimen2\font}%
    }%

    \Wrappedbreaksatspecials
    \OriginalVerbatim[#1,codes*=\Wrappedbreaksatpunct]%
}
\makeatother
\usepackage{thm-restate}

\newtheorem{open}{Open Question}

\makeatletter
\renewcommand*{\@fnsymbol}[1]{%
  \ifcase#1%
    \or \ensuremath{\dagger}
    \or \ensuremath{\ddagger}
    \or \ensuremath{\S}
    \or \ensuremath{\P}
    \or \ensuremath{\|}
    \else \@ctrerr%
  \fi
}
\makeatother

\begin{document}
\title{Can outcome communication explain Bell nonlocality?}

\begingroup
\renewcommand\thefootnote{*}
\footnotetext{These authors contributed equally to this work.}
\endgroup

\author{Carlos Vieira\textsuperscript{*}}%
\email{carlosvieira@ime.unicamp.br}
\affiliation{\unicamp}

\author{Carlos de Gois\textsuperscript{*}}%
\email{carlos.belini-de-gois@inria.fr}
\affiliation{\siegen}
\affiliation{\inria}

\author{Pedro Lauand}
\affiliation{\perimeter}
\affiliation{\waterloo}

\author{Lucas~E.~A.~Porto}
\affiliation{\sorbonne}

\author{Sébastien Designolle}
\affiliation{Zuse Institute Berlin, 14195 Berlin, Germany}
\affiliation{Inria, ENS de Lyon, UCBL, LIP, 69342, Lyon Cedex 07, France}

\author{Marco T\'{u}lio Quintino}
\email{marco.quintino@lip6.fr}
\affiliation{\sorbonne}

\begin{abstract}
    A central aspect of quantum information is that correlations between spacelike separated observers sharing entangled states cannot be reproduced by local hidden variable (LHV) models, a phenomenon known as Bell nonlocality.
    If one wishes to explain such correlations by classical means, a natural possibility is to allow communication between the parties.
    In particular, LHV models augmented with two bits of classical communication can explain the correlations of any two-qubit state.
    Would this still hold if communication is restricted to measurement outcomes?
    While in certain scenarios with a finite number of inputs the answer is yes, we prove that if a model must reproduce all projective measurements, then for any qubit–qudit state the answer is no. In fact, a qubit-qudit under projective measurements admits an LHV model with outcome communication if and only if it already admits an LHV model without communication. On the other hand, we also show that when restricted sets of measurements are considered (for instance, when the qubit measurements are in the upper hemisphere of the Bloch ball), outcome communication does offer an advantage. This exemplifies that trivial properties in standard LHV scenarios, such as deterministic measurements and outcome-relabelling, play a crucial role in the outcome communication scenario.
    \end{abstract}

\maketitle

\section{Introduction}
In a seminal work, John Bell proved that the correlations between spatially separated experiments predicted by quantum theory cannot be reproduced by classical models satisfying a natural notion of locality~\cite{bell_einstein_1964}. 
This phenomenon, known as Bell nonlocality, constitutes one of the basis of our current understanding of nature, and it is a central aspect of quantum information theory \cite{GENOVESE2005319,Brunner:2014RMP, Hensen:2015NAT, Giustina:2015PRL, Shalm:2015PRL}.

The class of classical models originally considered by Bell, termed local hidden variable (LHV) models, represents the behaviour of classical systems that have interacted in the past, but are then physically separated and unable to communicate [\cref{fig:Bell}].
As the model is not powerful enough to explain all quantum correlations, a natural question is to investigate what additional resources one would need in order to reproduce them.
If the observers are allowed to freely communicate, for instance, any nonsignalling behaviour --- hence, any quantum behaviour --- can be reproduced by classical models.

In the past few decades, many efforts have been dedicated to understanding what is the minimal amount of communication needed to achieve quantum correlations from classical strategies [\cref{fig:BellComm}].
Early efforts \cite{maudlin1992bell,brassard_cost_1999,steiner2000towards} and a series of improvements \cite{pati2000minimum,massar2001classical,cerf2000classical}
led to a proof that one classical bit is sufficient to reproduce the statistics of a maximally entangled two-qubit state under projective measurements, and that two bits of communication are sufficient to reproduce projective measurements on arbitrary two-qubit states~\cite{toner_communication_2003}. 
Later, it was also shown that two bits of communication are sufficient to classically simulate arbitrary positive operator-valued measure (POVM) measurements on any two-qubit state~\cite{renner_classical_2023}, and it is possible that even a single bit of communication is enough to do so~\cite{renner_minimal_2023}.

In this work, we take a different approach, and instead of investigating the minimal amount of communication needed to reproduce all quantum correlations, we fix \textit{what} can be communicated by the parties and explore which quantum correlations can then be reached \cite{chaves_unifying_2015, brask_bell_2017}.
In the first place, if one party is allowed to communicate their measurement setting, then all nonsignalling correlations can be reproduced by classical models \cite{brassard_cost_1999}.
On the other hand, if only the \textit{measurement outcome} is communicated by one party to another [\cref{fig:BellOutComm}], the answer becomes less obvious.

Such local hidden variable models augmented with outcome communication were investigated in scenarios with a finite number of measurement choices \cite{chaves_unifying_2015,brask_bell_2017}.
In this case, they can be significantly more powerful than the local model and, in particular, can reproduce any nonsignalling behaviour in the minimal scenario, including a Popescu-Rohrlich box \cite{Popescu:1994FPH}.
While previous research has analysed the power of outcome communication at the level of quantum behaviours, much less is known about models that aim to reproduce the statistics of all possible measurements on a bipartite quantum state.

\begin{figure*}[t]\label{fig:eapm-behaviour s}
    \centering
    \subfigure[Standard Bell scenario]{\label{fig:Bell}\includegraphics[width=.28\linewidth]{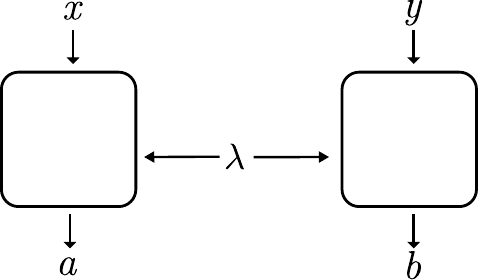}}\hspace{2em}
    \subfigure[Bell scenario with general communication]{\label{fig:BellComm}\includegraphics[width=.28\linewidth]{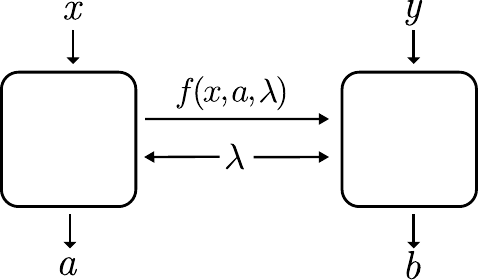}}\hspace{2em}
    \subfigure[Bell scenario with outcome communication]{\label{fig:BellOutComm}\includegraphics[width=.28\linewidth]{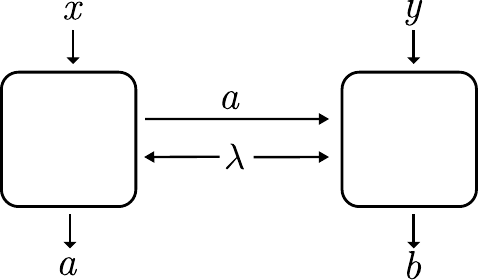}}
    \caption{
      Different causal structures for bipartite Bell scenarios.
      (a) Standard Bell scenario, where the parties are correlated solely through a shared hidden variable $\lambda$ and no communication occurs.
      (b) Bell scenario with general communication from Alice to Bob, where Alice’s message $f(x, a, \lambda)$  may depend on both her input and outcome and the hidden variable $\lambda$.
      (c) Bell scenario with outcome communication, where Alice’s message to Bob is restricted to her measurement outcome.
    }
\end{figure*}

This motivates a fundamental question: how well can local hidden variable models augmented with outcome communication reproduce all correlations of a quantum state?
Surprisingly, we show that when the communicating party performs arbitrary two-outcome measurements, a state is local if and only if it is local with communication of outcomes. As a special case, outcome communication cannot explain the correlations of any nonlocal qubit-qudit state under projective measurements.
Notably, our proof of this result strongly depends on the requirement that the local model with communication must also reproduce the correlations generated when the communicating party performs a deterministic measurement, i.e., a measurement that always returns the same outcome regardless of the underlying state.
This naturally raises the question of whether the result still holds if such a measurement is not considered.
To explore this, we analyse the case of two-qubit Werner states, which reveals particularities of outcome communication models that are not seen in standard Bell scenarios.

\section{The outcome communication model}
    We start by introducing the scenario and notation.
    A bipartite Bell scenario consists of two experimenters, Alice and Bob, who can locally measure their systems \cite{Bell:1964PHY, CHSH-inequality}.
    In each run of the experiment, they independently choose their measurements, labelled respectively by $x$ and $y$, and collect the corresponding outcomes, denoted by $a$ and $b$.
    By repeating the experiment multiple times and later gathering their data, they can infer the conditional probabilities $p(ab|xy)$.
    We refer to the set of probabilities $ p(ab|xy)$ for all $a, b, x$ and $y$ as the \textit{behaviour} of the system.

    Different models have been proposed to explain the potential correlations observed in the experiment's statistics under various assumptions \cite{brassard_cost_1999, toner_communication_2003, renner_classical_2023, chaves_unifying_2015, hall_relaxed_2011, brask_bell_2017, vieira2024test}.
    When Alice and Bob are free to choose their inputs and no communication occurs during the experiment, a natural model emerges, called the LHV model \cite{Brunner:2014RMP}.

    Within this model, since communication is forbidden, all correlations between the parties must be due to a common past, which is represented by a hidden variable $\lambda$ [see \cref{fig:Bell}].
    Conditioned on this common past, the observed statistics must factorise.
    More precisely, the behaviour $p(ab|xy)$ has a local hidden variable model (i.e., it is an LHV behaviour) if there is an additional variable $\lambda$, along with probability distributions $p(\lambda)$, $p_{A}(a|x\lambda)$ and $p_{B}(b|y \lambda)$, such that
    \begin{equation}
        p(ab|xy) = \sum_\lambda p(\lambda) p_{A}(a|x \lambda) p_{B}(b|y \lambda), \,\forall a,b,x,y .
    \label{eq:lhv-model}
    \end{equation}
    When a behaviour admits an LHV model, we say that this behaviour is Bell local, or simply, local.
    These behaviours satisfy the nonsignalling conditions, that is, they are such that $\sum_a p(ab|xy) = \sum_a p(ab|x^\prime y)$ and $\sum_b p(ab|xy) = \sum_b p(ab|xy^\prime)$.

    As is well known, there are behaviours obtained through local measurements on entangled quantum states that do not admit decomposition by a local hidden variable model.
    These are called nonlocal behaviours \cite{Bell:1964PHY, CHSH-inequality}.
    A natural attempt to better understand quantum correlations is thus to consider less restrictive models \cite{chaves_unifying_2015, hall_relaxed_2011, brask_bell_2017, vieira2024test}.

    Our focus is on the model where, in addition to the common past $\lambda$, correlations can arise from Alice communicating her outcomes to Bob [\cref{fig:BellOutComm}].
    \begin{definition}[LHV+Out model]
    \label{def:lhv-out-model}
        A behaviour $p(ab|xy)$ admits an LHV+Out model (or a local hidden variable model with communication of outcomes) when there is an additional variable $\lambda$, along with probability distributions $p(\lambda)$, $p_{A}(a|x\lambda)$ and $p_{B}(b|a y \lambda)$, such that
    \begin{equation}
        p(ab|xy) = \sum_\lambda p(\lambda) p_{A}(a|x \lambda) p_{B}(b|ay \lambda), \,\forall a,b,x,y .
    \label{eq:out-lhv-model}
    \end{equation}
    \end{definition}
    Trivially, all local behaviours are also LHV+Out, but the converse is not true.
    In particular, any local behaviour is nonsignalling, but the LHV+Out model includes some signalling behaviours.
    Less obviously, one can also find nonsignalling LHV+Out behaviours that are not LHV.
    Consider, for instance, the Popescu-Rohrlich (PR) boxes \cite{Popescu:1994FPH} in the paradigmatic Clauser-Horne-Shimony-Holt (CHSH) scenario \cite{CHSH-inequality}, which are known to be (extremely) nonlocal.
    The PR-box is a nonsignalling behaviour where the outcomes and inputs respect the relation $a\oplus b = xy$, where $\oplus$ is the modulo two bit addition.
    To see that they are LHV+Out, let $\lambda \in \set{0,1}$, $p(\lambda) = 1/2$, $p_{A}(a|x, \lambda) = \delta(a,x+\lambda)$, $p_{B}(b|a,y,\lambda) = \delta(b, (a \oplus \lambda)(y\oplus 1) \oplus \lambda)$.
    Then, it follows that
   $     \sum_{\lambda = 1}^2 p(\lambda)p_{A}(a|x \lambda) p_{B}(b|a y \lambda) = \frac{1}{2}\delta(a\oplus b, xy)$. 
    In fact, in the CHSH scenario, all nonsignalling behaviours are LHV+Out \cite{brask_bell_2017}, showing that outcome communication is a very powerful resource in some Bell scenarios.
    However, in general scenarios, not all nonsignalling behaviours are LHV+Out.
    It is even possible to obtain quantum behaviours that are not LHV+Out~\cite{chaves_unifying_2015}, a fact that was also verified experimentally~ \cite{ringbauer_experimental_2016}. 
    \cref{fig:behaviours-sets} illustrates the relationships between these sets of behaviours.
    \begin{figure}
        \centering
        \includegraphics[width=.7\columnwidth]{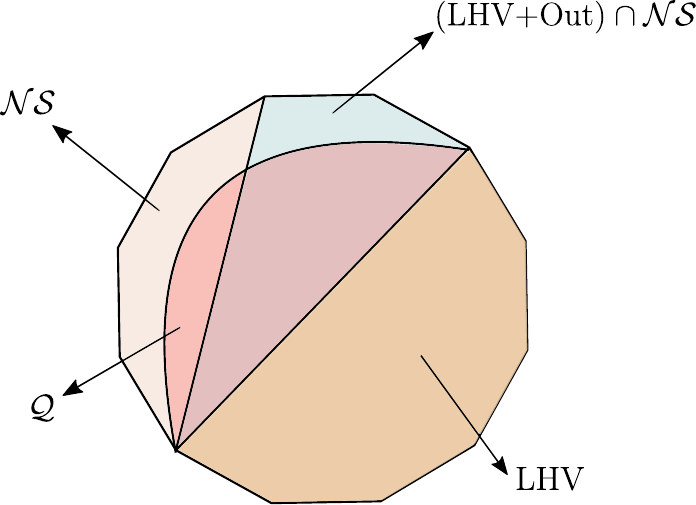}
        \caption{
          Pictorial representation of the relations between the different sets of correlations: $\text{LHV}$ denotes the of local correlations; $\mathcal{Q}$ the set of quantum correlations; $\mathcal{NS}$ the set of nonsignalling correlations; and $\text{(LHV+Out)} \cap \mathcal{NS}$ the set of local correlations augmented with outcome communication, restricted to the nonsignalling space.
          Note that $\text{LHV} \subset \mathcal{Q}$ and $ \text{LHV} \subset \text{(LHV+Out)} \cap \mathcal{NS}$.
          However, in general, $\mathcal{Q}$ and $\text{(LHV+Out)}$  $\cap$ $\mathcal{NS}$ are not subsets of one another.
        }
        \label{fig:behaviours-sets}
    \end{figure}

        Similarly to the case of Bell nonlocality, the set of LHV+Out behaviours is a polytope.
        Therefore, it is fully characterised by a finite number of extremal points, which are the deterministic strategies.
        A behaviour $p(ab|xy)$ is LHV+Out (see \cref{def:lhv-out-model}) if and only if it is a convex mixture of local deterministic processes with communication of outcomes:
        \begin{equation}\label{eq:out-lhv-deterministic-strategies-deltas}
            p(ab|xy) =  \sum_{\lambda} p(\lambda)\, D_A(a|x\lambda)D_B(b|ay\lambda),
        \end{equation}
        where $D_A(\,\cdot\, |x\lambda)$ and $D_B( \,\cdot \,|ay\lambda)$ are deterministic probability distributions.
        The proof is analogous to that of Fine's theorem \cite{Fine:1982PRL}.
    
    The notion of LHV behaviour [\cref{eq:lhv-model}] can be passed on to quantum states:
    A bipartite state $\rho$ admits an LHV model if, for any sets of local measurements $\set{A_{a|x}}$ and $\set{B_{b|y}} $, the behaviour with probabilities $p(ab|xy) = \tr\left[ \left( A_{a|x} \otimes B_{b|y} \right)  \rho \right]$ has an LHV model [\cref{eq:lhv-model}], where $\set{A_{a|x}}$ is a measurement if $A_{a|x} \succcurlyeq 0$ for any $a,x$, with $\sum_a A_{a|x} = \eye$ for any $x$, and analogously for $\set{B_{b|y}} $.
    Characterising the set of states that admits an LHV model is a central goal in nonlocality and the topic of an extensive literature, e.g., \cite{werner_quantum_1989, acin2006grothendieck, cavalcanti2016general, hirsch2016algorithmic, hirsch2017betterlocalhidden, DIB+23, augusiak2014review}.
    We extend this definition to the case of the LHV+Out model.
    \begin{definition}[LHV+Out states]
        A bipartite quantum state $\rho$ admits an LHV+Out model for the set of measurements $\set{ A_{a|x} }$ and $\set{ B_{b|y} }$ when the behaviour $p(ab|xy) = \tr[ (A_{a|x} \otimes B_{b|y}) \rho]$ admits an LHV+Out model [\cref{eq:out-lhv-model}].
        We say that $\rho$ admits an LHV+Out model when it admits an LHV+Out model for all possible sets of measurements $\set{ A_{a|x} }$ and $\set{ B_{b|y} }$.
    \label{def:out-lhv-states}
    \end{definition}

    As is usual in nonlocality, it is often relevant to analyse the case where the sets of $\set{ A_{a|x} }$ and $\set{ B_{b|y} }$ are the sets of all projective measurements, i.e., when $A_{a|x} A_{a^\prime|x} = \delta_{a, a^\prime} A_{a|x}$ and $B_{b|y} B_{b^\prime|y} = \delta_{b, b^\prime} B_{b|y}$ for any $x,y$.
    We remark that, as of today, there is no example where non-projective measurements, i.e., positive operator-valued measurements (POVM), are necessary for nonlocality \cite{RennerCompatibility2024}.

    \begin{remark}
        It follows from the definitions that any state that admits an LHV model also admits an LHV+Out model for any set of measurements.
        On the other hand, since the LHV+Out model for behaviours is generally more powerful than the LHV model, there could exist nonlocal states that admit an LHV+Out model.
        This should perhaps even be expected, as the LHV+Out model allows for a considerable amount of signalling, and will be further discussed at a later point.
    \end{remark}

    To show that a given behaviour $p(ab|xy)$ admits an LHV+Out model, it is enough to check whether $p(ab|xy)$ is inside the polytope given by the convex hull of \cref{eq:out-lhv-deterministic-strategies-deltas}.
    On the other hand, showing that a state $\rho$ admits an LHV+Out model is a much more challenging problem, since one has to consider all possible measurements that may be performed by Alice and Bob.
    This difference is analogous to the standard LHV scenario where no communication is allowed \cite{cavalcanti2016general,hirsch2016algorithmic, hirsch2017betterlocalhidden, DIB+23}.

\section{Outcome communication cannot explain nonlocality}
    One of our main results in this section shows the rather counter-intuitive fact that outcome communication provides no help in explaining the statistics of projective measurements on qubit-qudit states.
    Put another way, it shows that a qubit-qudit state has an LHV+Out model for projective measurements if and only if it has an LHV model.
    
    Before presenting this result, we first address a more elementary question: under which conditions can we construct an LHV model from an LHV+Out model? We formalise this in the following theorem, which will serve as a key tool for proving our main result.
    \begin{restatable}{theorem}{resthm}\label{thm:LHV+Out_Eq_LHV_Behaviours}
        Let $p(ab|xy)$ be a nonsignalling behaviour where Alice's outcomes are dichotomic $(a \in \set{\pm1})$.
        Suppose that $p(ab|xy)$ admits an LHV+Out model and that there exists an input $x'$ such that $p_{A}(\, \cdot \,|x')$ is deterministic.
        Then, $p(ab|xy)$ admits an LHV model.
    \end{restatable}
    The proof of \cref{thm:LHV+Out_Eq_LHV_Behaviours} is provided in \cref{sec: appendix proof Theorem}.
    We now turn to our main result.
    \begin{result}\label{res:MainResult}
    Let $\set{A_{a|x}}$ be a set of dichotomic measurements containing the deterministic measurement $\set{\eye,0}$, and $\set{B_{b|y}}$ be an arbitrary set of measurements. A bipartite quantum state $\rho\in\mathcal{L}(\mathcal{H}_A\otimes\mathcal{H}_B)$ admits an LHV+Out model for the measurements $\set{A_{a|x}}$ and $\set{B_{b|y}}$ if and only if  $\rho$ admits an LHV model for the measurements $\set{A_{a|x}}$ and $\set{B_{b|y}}$.
        
        In particular, if $\rho \in \mathcal{L}(\mathbb{C}^2\otimes\mathbb{C}^d)$ is a qubit-qudit state, then under projective measurements, $\rho$ admits an LHV model if and only if it admits an LHV+Out model.
    \end{result}
    The proof of \cref{res:MainResult}  follows directly from \cref{thm:LHV+Out_Eq_LHV_Behaviours}. When Alice performs the deterministic measurements, her marginal is necessarily deterministic, regardless of the quantum state $\rho$. Hence, if the associated behaviour admits an LHV+Out model, \cref{thm:LHV+Out_Eq_LHV_Behaviours} ensures that it must also admit an LHV model. As for the second part of \cref{res:MainResult}, notice that the set of all projective measurements includes the deterministic measurement $\set{\eye,0}$, and for qubits, projective measurements are necessarily dichotomic.

    Although the LHV+Out model is signalling and can even simulate PR boxes in the CHSH scenario, our main result shows that it offers no advantage over purely local models for qubit-qudit states.
    In contrast, if we allow the transmitted bit to be a function $f(x, a, \lambda) \in \set{0, 1}$ of Alice's input and the shared variable [see \cref{fig:BellComm}], instead of being directly the outcome, then any behaviour of a two-qubit maximally entangled state can be reproduced \cite{toner_communication_2003, renner_classical_2023}.
    This shows an extreme difference in the advantage that these two types of communication allow in the task of simulating quantum correlations.

    A curious feature of \cref{res:MainResult} is the prominent role played by a deterministic quantum measurement $\set{\eye, 0}$. At an intuitive level, the combination of such a dichotomic deterministic measurement with the nonsignalling constraint severely restricts the structure of outcome-dependent strategies in an LHV+Out model. Indeed, if Alice has a deterministic measurement, then for this input, her outcome is fixed and therefore is uncorrelated with the hidden variable. Together with the assumption of nonsignalling and the dichotomic nature of Alice’s outcomes, this forces Bob’s response functions to become effectively independent of Alice’s outcome [see \cref{eq:p_B_Model}]. As a result, any LHV+Out model compatible with these constraints can be reduced to a standard LHV model.

    To our knowledge, no other example in the literature of Bell nonlocality assigns such a crucial role to deterministic measurements.
    In fact, in standard LHV models for quantum states, it is straightforward to account for a deterministic measurement by coarse-graining the outcomes of a nondeterministic one.
    More fundamentally, deterministic measurements usually are not even considered genuinely quantum, since they can be implemented by simply discarding the quantum system and outputting the outcome that occurs with unit probability.
    This raises the question of whether deterministic measurements are really the underlying factors leading to the equivalence between LHV and LHV+Out states, or whether a more subtle condition can be identified.
    
    It is well known that the extremal qubit dichotomic measurements are the deterministic measurements $\set{\eye, 0}$, $\set{0, \eye}$, and rank-1 projective measurements which can be represented in the Bloch sphere $M_{\pm} = \frac{\eye \pm \vec{n}\cdot \vec{\sigma}}{2}$ \cite{dariano2005extremal}.
    It is thus natural to investigate whether a qubit-qudit state which is LHV+Out with respect to rank-1 projective measurements but not with respect to the deterministic measurements can violate a Bell inequality.
    We show that this is not the case for the arguably most natural candidate two-qubit states, namely, the two-qubit Werner states $W(v) = v \dyad{\psi^-} + (1-v)\tfrac{\eye}{4}$, where $\ket{\psi^-} = (\ket{01} - \ket{10})/\sqrt{2}$ and $v \in [0,1]$.

    \begin{restatable}{result}{werneroutlhv}\label{res:werner-outlhv}
        For the set of all rank-1 projective measurements, 
        a Werner state $W(v)$ admits an LHV+Out model 
        if and only if it admits an LHV model.
    \end{restatable}
    
    This result follows from an adaptation of Ref.~\cite[Lemma II.1]{divianszky_certification_2023} (see \cref{app:proof-result-2} for details). 
    In contrast, if Alice's measurements are restricted to lie on a single hemisphere of the Bloch sphere, the situation changes: LHV and LHV+Out models for the Werner state exhibit a distinct separation.

    \begin{result} \label{res:lhvout_dif_lhv}
        If Alice's measurements are all projective measurements lying on the upper hemisphere of the Bloch sphere, the Werner state $W(v)$ admits an LHV+Out model for $v \leq 0.69828$.
        Meanwhile, it is known to violate a Bell inequality for $v > 0.69604$ \cite{designolle2024better}. 
    \end{result}
    Curiously, in the LHV case, having a model for projective measurements within a single hemisphere already suffices to extend it to all dichotomic measurements, since antipodal measurements are equivalent via outcome relabelling. \cref{res:lhvout_dif_lhv} shows this is not the case for LHV+Out models.

    To construct this LHV+Out model, we extend computational techniques for studying local models (namely, polytope approximations~\cite{cavalcanti2016general, hirsch2016algorithmic, hirsch2017betterlocalhidden, gois2021general, DIB+23} and Frank-Wolfe algorithms~\cite{BCC+22, DIB+23, besancon2025improved}) to the LHV+Out scenario. These techniques are discussed in \cref{app:proof-result-3}, and an explicit model is provided in a repository \cite{zenodo_sm}.
    
    The above results suggest that a key structural requirement behind the equivalence between LHV and LHV+Out models is the presence of antipodal measurements. In particular, whenever an LHV+Out model is required for a measurement $x = \set{A_0, A_1}$, if one simultaneously demands a model for the swapped measurement $x' = \set{A_1, A_0}$, the additional freedom provided by outcome-dependent strategies appears to become ineffective. Intuitively, this outcome-inversion symmetry constrains Bob’s response functions to be compatible with both orientations of the same measurement, preventing them from exploiting Alice’s outcome as an independent source of information. This observation motivates the open question formulated below, which aims to characterise whether such antipodal symmetry is sufficient, at the level of nonsignalling behaviours, to ensure that LHV+Out models always reduce to standard LHV ones.

    \begin{open}\label{open_prob}
        Let $p(ab|xy)$ be a nonsignalling behaviour where $x \in \set{0, ..., 2m_x-1}$, $a \in \set{\pm1}$ and $p(ab|xy) = p(-a, b|x + m_x, y)$, with addition modulo $2m_x$. If $p(ab|xy)$ admits an LHV+Out model, does it also admit an LHV model?
    \end{open}

    Our findings provide evidence pointing towards an affirmative answer to the above question.
    First, in the simplest case where Bob's measurements are also dichotomic and both Alice's and Bob's marginals are uniform, the discussion presented in \cref{app:proof-result-2} establishes that the answer is indeed affirmative.
    Second, using linear programming, we computationally verified the same conclusion in scenarios where $m_x \leq 4$.
    Finally, we have sampled a large set of random behaviours respecting the hypothesis of the open question, and no counterexample was found.
   
\section{Discussion}
    In this work, we defined the notion of LHV+Out models for quantum states.
    We then established general conditions under which a behaviour that admits an LHV+Out model also admits an LHV model (\cref{thm:LHV+Out_Eq_LHV_Behaviours}), which allowed us to identify scenarios in which the LHV and LHV+Out models are equivalent at the level of quantum states.
    In particular, we showed that when Alice's measurements are dichotomic, a quantum state admits an LHV model if and only if it admits an LHV+Out model.
    As a consequence, any qubit–qudit state under projective measurements satisfies this equivalence (\cref{res:MainResult}).
    Notably, the equivalence established in \cref{res:MainResult} relies on the presence of a deterministic measurement on Alice. This naturally led us to investigate whether such deterministic measurements are indeed necessary for the equivalence between LHV+Out and LHV models.
    For the case of two-qubit Werner states, we showed that for rank-1 projective measurements the equivalence holds regardless (\cref{res:werner-outlhv}). However, restricting Alice’s measurements to a single hemisphere of the Bloch sphere leads to a strict separation between the two models (\cref{res:lhvout_dif_lhv}). 
    
    These findings suggest that the key requirement behind the equivalence is the presence of antipodal measurements, precisely formulated in \cref{open_prob}. We emphasise that, should the \cref{open_prob} be answered affirmatively, the equivalence established in \cref{res:werner-outlhv} would extend beyond Werner states and hold for arbitrary qubit–qudit states under projective measurements.

    From the proof techniques used, we suspect that \cref{res:MainResult} strongly relies on the fact that one party is restricted to dichotomic measurements.
    Also, note that a quantum measurement with two outcomes is completely described by a single measurement operator, since the other one is implicitly defined by the measurement normalisation condition.
    This is not the case for measurements with more than two outcomes, and this particularity of dichotomic measurements is one of the reasons for results like Gleason's theorem~\cite{Gleason:1957JMM} not to hold on the dichotomic case.
    In this regard, it is important to stress that \cref{thm:LHV+Out_Eq_LHV_Behaviours} does not extend in a straightforward manner beyond the dichotomic setting. In fact, one can construct simple counterexamples by embedding a PR-box~\cite{Popescu:1994FPH} into a larger outcome space and adding an extra measurement that is fully deterministic on the newly introduced outcome. The resulting behaviour admits an LHV+Out model and contains a deterministic measurement, yet it cannot admit an LHV model, as this would imply locality of the underlying PR correlations. Therefore, we expect LHV+Out models beyond dichotomic measurements to be a rich and interesting area to be explored, which should help us to understand quantum correlations in communication scenarios.

    Finally, an interesting observation is that marginalising over Alice's outcome in an LHV+Out model yields a classical model in the context of entanglement-assisted classical communication (EACC) scenarios \cite{tavakoli2021correlations, pauwels2021entanglement, Frenkel2022entanglement,Moreno_Semi_device_2021}.
    Therefore, the techniques and results developed in this work may find further applications in the study of EACC protocols, particularly in understanding the classical simulation of quantum correlations assisted by entanglement \cite{vieira_interplays_2023, adaptive2022pauwels}.

\section*{Acknowledgments}
We thank C.~Duarte, M.~Renner, B.~Rizzuti, S.~Schlösser, T.~Vértesi, E.~Wolfe for fruitful discussions.
This work was supported by the São Paulo Research Foundation (FAPESP) under grants No. 2024/16657-3 and 2025/01058-0, by the Brazilian National Council for Scientific and Technological Development (CNPq) under grant no.~445328/2024-0, by \siegenfinancialsupport the European Union's Horizon 2020 Research and Innovation Programme under QuantERA Grant Agreement no.\ 731473 and 101017733, and QuantEdu France, a State aid managed by the French National Research Agency for France 2030 with the reference ANR-22-CMAS-0001. MTQ acknologes the funding Tremplins nouveaux entrants \& nouvelles entrantes - Edition 2024, project HOQO-KS.
Research at Perimeter Institute is supported by the Government of Canada through the Department of Innovation, Science and Economic Development Canada and by the Province of Ontario through the Ministry of Research, Innovation and Science.
Research reported in this paper was partially supported through the Research Campus Modal funded by the German Federal Ministry of Education and Research (fund numbers 05M14ZAM, 05M20ZBM) and the Deutsche Forschungsgemeinschaft (DFG) through the DFG Cluster of Excellence MATH+.

\bibliographystyle{quantum}
\bibliography{1_bibliography}

\onecolumn
\appendix
\crefalias{section}{appendix}

\section{Proof of Theorem \ref{thm:LHV+Out_Eq_LHV_Behaviours}}\label{sec: appendix proof Theorem}
\resthm*
\begin{proof}
        Without loss of generality, assume that $p_A(1|x') = 1$.
        The case $p_A(-1|x') = 1$ is analogous.

        Since $p(ab|xy)$ admits an LHV+Out model, we have
        \begin{equation}\label{eq:LHV+OutModelProof}
             p(ab|xy) = \sum_\lambda p(\lambda) p_{A}(a|x \lambda) p_{B}(b|ay \lambda), \,\forall a,b,x,y .
        \end{equation}
        We claim that $p(ab|xy)$ admits the following LHV model:
        \begin{equation}\label{eq:LHVModelProof}
             p(ab|xy) = \sum_\lambda p(\lambda) p_{A}(a|x \lambda) \tilde{p}_{B}(b|y \lambda), \,\forall a,b,x,y .
        \end{equation}
        where $p(\lambda)$ and $p_A(a|x\lambda)$ are the same as in \cref{eq:LHV+OutModelProof}, and $\tilde{p}_{B}(b|y \lambda) :=  p_{B}(b|1y \lambda)$ $\forall b, y, \lambda$.
        For $a=1$, \cref{eq:LHV+OutModelProof} directly implies \cref{eq:LHVModelProof}.
        Thus, it remains to verify the case $a=-1$.
        Marginalising \cref{eq:LHV+OutModelProof} over Bob's outcomes, and setting $a=1$, $x=x'$, we obtain
        \begin{equation}
            1 = p_{A}(1|x') = \sum_\lambda p(\lambda) p_{A}(1|x' \lambda).
        \end{equation}
        Hence, for all $\lambda$, we must have $p_A(1|x'\lambda) = 1$ and consequently $p_A(-1|x'\lambda) = 0$.
        
        By the nonsignalling property,
        \begin{equation}\label{eq:NonSigConditionProof}
            p(-1b|xy) = p_{B}(b|y) - p(1b|xy).
        \end{equation}
        On the other hand, marginalising \cref{eq:LHV+OutModelProof} over Alice's outcomes at input $x'$, we find
        \begin{align}\label{eq:p_B_Model}
             p_{B}(b|y) &:= \sum_{a}p(ab|x'y) = \sum_{a,\lambda} p(\lambda) p_{A}(a|x' \lambda) p_{B}(b|ay \lambda) \nonumber\\
             & = \sum_{\lambda} p(\lambda) p_{B}(b|1y \lambda).
        \end{align}
        Combining \cref{eq:NonSigConditionProof,eq:LHV+OutModelProof,eq:p_B_Model}, we have
        \begin{align}
            p(-1b|xy) &= p_{B}(b|y) - p(1b|xy) \nonumber\\
            &= \sum_{\lambda} p(\lambda) \left[ 1 - p_{A}(1|x\lambda) \right]p_{B}(b|1y \lambda) \nonumber\\
            &= \sum_{\lambda} p(\lambda) p_{A}(-1|x\lambda)\tilde{p}_{B}(b|y \lambda).
        \end{align}
        which matches the model in \cref{eq:LHVModelProof} for $a=-1$.
    \end{proof}

\section{Proof of Result \ref{res:werner-outlhv}}
\label{app:proof-result-2}

    \werneroutlhv*
    While Result 2 was stated for the specific case of Werner states, the next theorem establishes that the same equivalence remains valid in a broader context. It applies to all two-qubit states with maximally mixed marginals, namely the Bell-diagonal family.
Result 2 follows from the following theorem.
    \begin{theorem}
        Let $\rho\in\mathcal{L}(\mathbb{C}_{2} \otimes\mathbb{C}_{2})$ be a two-qubit state with uniform partial states, that is, $\tr_A(\rho)= \tr_B(\rho)=\frac{\eye}{2}$.
         Under rank-1 projective measurements, $\rho$ admits an LHV model if and only if it admits an LHV+Out model.
    \end{theorem}

    \begin{proof}[Proof]
The statement is proven by showing that if $\rho$ violates a Bell inequality under rank-1 projective measurements, then it also violates an LHV+Out inequality (the converse follows from \cref{def:lhv-out-model}).
First, recall that any two-qubit state violating a Bell inequality under projective measurements also violates a dichotomic Bell inequality (all extremal dichotomic measurements are projective~\cite{dariano2005extremal}). Moreover, since $\rho$ has uniform marginals, we have $\ev{A_x}_{\rho} = \ev{B_y}_{\rho} = 0$ for any observables $A_x = A_{0|x}-A_{1|x}$ and $B_y = B_{0|y}-B_{1|y}$, where $\set{A_{a|x}}, \set{B_{b|y}}$ are rank-1 projective measurements. Hence, any violation can be expressed through a full-correlator Bell inequality of the form
\begin{equation}
    \mathcal{O}_{\text{LHV}} = \sum_{x=1}^{m} \sum_{y=1}^{n} M_{xy}\,\ev{A_x B_y} \leq L(M),
    \label{eq:fc-bell-inequality-proofsketch}
\end{equation}
where $M=\set{M_{xy}}$ is a real coefficient matrix.
The local (LHV) bound is
\begin{equation}
    L(M) \coloneqq \max_{\substack{\vec a \in \{-1,1\}^m \\ \vec b \in \{-1,1\}^n}}
    \sum_{x=1}^m \sum_{y=1}^n M_{xy}\,a_x b_y.
    \label{eq:Bound_FCLHV}
\end{equation}

If outcome communication from Alice to Bob is allowed, Bob's deterministic response may depend on both his setting $y$ and Alice's outcome $a_x$. The corresponding LHV+Out bound is
\begin{equation}
    L_{\mathrm{Out}}(M)\coloneqq
    \max_{\substack{\vec a \in \{-1,1\}^m \\ \vec{b}\in\{-1,1\}^{n \times 2}}}
    \sum_{x=1}^m\sum_{y=1}^n M_{xy}\,a_x\,b_{y,a_x}.
    \label{eq:Bound_FCOut}
\end{equation}
We can manipulate the expression of $L_{\text{Out}}(M)$ to get: 
\small
        \begin{align}
          & L_{\text{Out}}(M) \coloneqq  \max_{\substack{\vec{a} \in \{-1,1\}^{n} \\ \vec{b} \in \{-1,1\}^{n \times 2}}} \sum_{x=1}^{m} \sum_{y=1}^{n} M_{xy} a_x  b_{y,a_x}  \nonumber\\
           &= \max_{\substack{\vec{a} \in \{-1,1\}^{n} \\ \vec{b} \in \{-1,1\}^{n \times 2}}} \left[\sum_{y=1}^{n} \sum_{x:a_x = +1} M_{xy} b_{y,+1} - \sum_{y=1}^{m} \sum_{x:a_x = -1} M_{xy} b_{y,-1} \right] \nonumber\\
           & = \max_{\substack{\vec{a} \in \{-1,1\}^{n} \\ \vec{b} \in \{-1,1\}^{n \times 2}}} \left[\sum_{y=1}^{n} b_{y,+1} \sum_{x:a_x = +1} M_{xy} - \sum_{y=1}^{m} b_{y,-1} \sum_{x:a_x = -1} M_{xy} \right] \nonumber\\
           & = \max_{\vec{a} \in \{-1,1\}^{n}} \left[\sum_{y=1}^{n} \left|\sum_{x:a_x = +1} M_{xy} \right| + \sum_{y=1}^{m} \left|\sum_{x:a_x = -1} M_{xy} \right| \right] \nonumber\\
           & = L_{2}(M),
    \end{align}
\normalsize
     where $L_2(M)$ is the classical bound of the associated prepare-and-measure scenario inequality introduced in Ref.~\cite{divianszky_certification_2023}. In particular, Lemma~II.1 therein shows that for the symmetrised matrix
    $$M'=\begin{pmatrix} M \\ -M \end{pmatrix},$$
    one has the identity $ L_2(M') = L(M') = 2L(M)$, thus for this case 
$   L_{\mathrm{Out}}(M') = 2L(M)$.
Therefore, the following Bell expression:
\begin{equation}
    \mathcal{O}_{\text{LHV+Out}} =
    \sum_{x=1}^m \sum_{y=1}^n M_{xy}\ev{A_x B_y}
    -\sum_{x=1}^m \sum_{y=1}^n M_{xy}\ev{A_{m+x} B_{y}} ,
    \label{eq:lhv-out-inequality-proofsketch}
\end{equation}
has LHV and LHV+Out bounds both equal $2L(M)$.

On the other hand, since $\rho$ violates \eqref{eq:fc-bell-inequality-proofsketch}, there exist observables such that
\begin{equation}
   Q_\rho(M)\coloneqq \sum_{x=1}^m\sum_{y=1}^n M_{xy}\,\ev{A_x B_y}_\rho > L(M).
\end{equation}
Defining the extended observables
\begin{equation}
   \tilde A_x = \begin{cases}
       A_x & 1\leq x\leq m, \\
       -A_{x} & m+1\leq x\leq 2m,
   \end{cases}
\end{equation}
we obtain
    \begin{align}
            \mathcal{O}_{\text{LHV+Out}} = \sum_{x=1}^{m} & \sum_{y=1}^{n} M_{xy} \ev{ \tilde{A}_x B_y }_{\rho} \nonumber\\
            & - \sum_{x=m+1}^{2m} \sum_{y=1}^{n} M_{xy} \ev{ \tilde{A}_x B_y }_{\rho} \nonumber\\
            & = 2\sum_{x=1}^{m}  \sum_{y=1}^{n} M_{xy} \ev{ A_x B_y }_{\rho} \nonumber \\
            &= 2Q_{\rho}(M) > 2L(M),
    \end{align}
  therefore violating the LHV+Out bound. In conclusion, if $\rho$ violates any Bell inequality [\cref{eq:fc-bell-inequality-proofsketch}], one can construct an LHV+Out inequality [\cref{eq:lhv-out-inequality-proofsketch}] that is also violated by $\rho$.
\end{proof}

\section{Proof of Result \ref{res:lhvout_dif_lhv}}
\label{app:proof-result-3}
To construct LHV+Out models for quantum states, we extend the computational methods originally developed for LHV models \cite{cavalcanti2016general, hirsch2016algorithmic, hirsch2017betterlocalhidden, gois2021general}. Before showing details of the model itself, let us outline the general argument behind the method.

Ultimately, we wish to make a statement about the behaviour of a state under all projective measurements in a hemisphere.
Because that is a continuous set, the core idea is to first work with an inner approximation of the set of quantum measurements in that hemisphere, composed of a finite but large number of measurements, then certify the corresponding behaviour is inside LHV+Out polytope, and finally extend it to a model for the entire convex hull of the original measurements.

Recent advances in conditional gradient algorithms \cite{BCC+22,besancon2022frankwolfe} and an efficient implementation have enabled the numerical search for LHV models for behaviours with up to hundreds of measurements \cite{DIB+23}.
We extended this method to the LHV+Out scenario \cite{bellpolytopes-lhvout}, the central modification being to rewrite Bob’s response functions so that they can explicitly depend on Alice’s outcome. Moreover, the efficiency of the implementation was also greatly improved by employing the symmetrisation techniques described in Ref.~\cite{Designolle2024Symmetric, besancon2025improved}. Importantly, models obtained with this method are self-contained, and can be verified independently of the algorithm that searches for the model (see \cref{app:numerics-verification}).

This method was then applied starting from a set of $401$ measurements per party, all restricted to the upper hemisphere of the Bloch sphere (i.e., the measurement Bloch vectors $\vec{u}_x, \vec{u}_y$ all lying on the upper hemisphere, for all measurement settings $x$ and $y$).
For this configuration, we found a model that very closely approximates the statistics of the Werner state with visibility $v = 0.7071$. Indeed, the Euclidean distance between the model’s statistics and the quantum ones was $\epsilon \approx 2 \times 10^{-4}$.
This model can be made exact by slightly reducing the visibility to $v = \tfrac{1}{1+\epsilon}\cdot 0.7071 \approx 0.70695$ (see Lemma I in Ref.~\cite{DIB+23}). The initial measurement set and the resulting model are available in an online repository \cite{zenodo_sm}.

Lastly, we must still argue that the model remains valid beyond the behaviour of $W(v)$ under the finite choice of measurements considered.
A well known method in the LHV case is to take the convex hull of the measurement vectors for each party, compute the radii $\eta_A$ and $\eta_B$ of the largest spheres that can be inscribed in each convex hull (also known as the \emph{shrinking factors}), then observe that the state \( W( \eta_A \eta_B v)\) admits the model for all rank-1 measurements \cite{cavalcanti2016general, hirsch2016algorithmic}. This argument extends directly to the LHV+Out case.
To enlarge the inscribed radius, a common trick in the LHV case is to double the measurement set by including all antipodal points on the Bloch sphere, that is, given a measurement $\vec{u}$, one also adds its antipode $-\vec{u}$. It is straightforward to show that if an LHV model exists for the sets $\set{u_x}_x$ and $\set{u_y}_y$, then it also exists for the enlarged sets $\set{u_x}_x \cup \set{-u_x}_x$ and $\set{u_y}_y \cup \set{-u_y}_y$.
However, it is crucial to notice that in the LHV+Out scenario this doubling trick does not apply symmetrically:
while it can be implemented for Bob’s measurements, it cannot be applied to Alice’s side: since Alice communicates her input to Bob, this restricts the pre-processing she can perform without altering Bob’s strategy.

Our initial $401$ measurements were chosen so that their convex hull provides a good approximation to the upper hemisphere of the Bloch sphere \cite{zenodo_sm}.
Indeed, the convex hull of this set contains a half-sphere of radius $\cos(\pi/40)^2 \approx 0.99384$. By further applying the doubling trick only to Bob’s measurements and applying the shrinking factors, we conclude that the Werner state $W(v)$ admits an LHV+Out model for all $v \leq  0.70695 \cos(\pi/40)^4 \approx 0.69828$, valid for arbitrary projective measurements on Bob’s side and for projective measurements restricted to the upper hemisphere on Alice’s side. Since $W(v_{\mathrm{NL}})$ is known to be nonlocal for any $v_{\mathrm{NL}} \geq 0.69604$ \cite{designolle2024better}, this establishes \cref{res:lhvout_dif_lhv}.

In the remainder of this section we present and validate the obtained LHV+Out model in more detail.

\subsection{Correlator representation}
\label{app:correlator-notation}

\noindent
We start by constructing a representation of LHV+Out behaviours in terms of correlators. This representation is more economical and ergonomic than the full probability representation, and will later simplify the construction of the model. 
The following proposition shows how to define an LHV+Out behaviour in terms of the correlators representation.

\begin{restatable}{proposition}{corrprop}\label{prop:LHV+Out-correlators}
    A nonsignalling behaviour $p(ab|xy)$ with binary outcomes has an LHV+Out model if, and only if, there exists a probability distribution $p(\lambda)$ and correlators $\set{\langle a_{x}^{\lambda}\rangle}_{x, \lambda}, \set{\langle b_{y,1}^{\lambda}\rangle}_{y, \lambda}, \set{\langle b_{y,-1}^{\lambda}\rangle}_{y, \lambda} \subseteq \set{\pm1}$ such that
    \begin{subequations}\label{eq:CorrelatorsOutLHV}
        \begin{align}
            &\langle a_x \rangle = \sum_\lambda p(\lambda) \langle a_{x}^{\lambda} \rangle \label{eq:CorrelatorsOutLHV-alice}\\
             &\langle b_y \rangle = \sum_\lambda p(\lambda) \langle b_{y, \langle a_{x}^{\lambda} \rangle}^{\lambda} \rangle,  \quad \forall x  \label{eq:CorrelatorsOutLHV-bob}\\
            & \langle a_x b_y \rangle = \sum_\lambda p(\lambda) \langle a_{x}^{\lambda} \rangle \langle b_{y, \langle a_{x}^{\lambda} \rangle}^{\lambda} \rangle .\label{eq:CorrelatorsOutLHV-joint}
        \end{align}
    \end{subequations}
\end{restatable}
Here, $\langle a_x \rangle$ and $\langle b_y \rangle$ are called the marginal correlators for Alice and Bob, respectively, while $\langle a_x b_y \rangle$ is the two-body correlator.

\begin{proof}
    Let start with the "$\Rightarrow$" direction.
    As $p(ab|xy)$ has an LHV+Out model, then there exists a probability distribution $p(\lambda)$ and deterministic probability distributions $ D_{A}(a|x \lambda)$ and $D_{B}(b|a y \lambda)$, such that
    \begin{equation*}
        p(a b|x y)=\sum_{\lambda} p(\lambda) D_{A}(a|x \lambda) D_{B}(b|a y \lambda).
    \end{equation*}
    where $D_{A}(a|x \lambda)$ and $D_{B}(b|a y \lambda)$ are deterministic.
    Let us define
    \begin{subequations}
    \begin{align}
        &\langle a_{x}^{\lambda} \rangle = \sum_{a} a D_{A}(a|x \lambda);\\
        &\langle b_{y, 1}^{\lambda} \rangle=\sum_{b} b D_{B}(b|1 y \lambda);\\				
        & \langle b_{y, -1}^{\lambda} \rangle = \sum_{b} b D_{B}(b|-1 y \lambda).
    \end{align}
    \end{subequations}
    It follows that $\langle a_{x}^{\lambda} \rangle, \langle b_{y, 1}^{\lambda} \rangle, \langle b_{y, -1}^{\lambda} \rangle \in \set{\pm1}$.
    Moreover, we also have that $D_{A}(a|x \lambda) = \delta(a, \langle a_{x}^{\lambda} \rangle)$, which implies
    \begin{equation}
        D_{A}(a|x \lambda) D_{B}(b|a y \lambda) = D_{A}(a|x \lambda) D_{B}(b|\langle a_{x}^{\lambda} \rangle y \lambda) .
    \end{equation}
    Thus,
    \begin{align}
        &\sum_{\lambda} p(\lambda)\langle a_{x}^{\lambda} \rangle \langle b_{y, \langle a_{x}^{\lambda} \rangle}^{\lambda} \rangle \nonumber\\
        &= \sum_{a,b} ab \sum_{\lambda} p(\lambda) D_{A}(a|x \lambda) D_{B}(b|\langle a_{x}^{\lambda} \rangle y \lambda) \nonumber\\
        & = \sum_{a,b} ab \sum_{\lambda} p(\lambda) D_{A}(a|x \lambda) D_{B}(b|a y \lambda) \nonumber\\
        &= \langle a_{x} b_{y} \rangle.
    \end{align}
    Moreover, it is also true that
    \begin{subequations}
    \begin{align}
     \langle a_{x} \rangle & = \sum_{a} a p(a|x)  =\sum_{\lambda} p(\lambda) \sum_{a} a D_{A}(a|x \lambda)  \\
    &= \sum_{\lambda} p(\lambda) \langle a_{x}^{\lambda} \rangle.
    \end{align}
    \end{subequations}
    Finally, we also have
    \begin{align}
    & \langle b_{y} \rangle = \sum_{b} b p(b|y) = \sum_{a,b} b p(ab|xy) \nonumber\\
    & =\sum_{\lambda} p(\lambda) \sum_{a,b} b D_{A}(a|x \lambda)  D_{B}(b|a y \lambda) \nonumber\\
    & =\sum_{\lambda} p(\lambda) \sum_{a,b} b  D_{A}(a|x \lambda) D_{B}(b|\langle a_{x}^{\lambda} \rangle y \lambda) \nonumber\\
    & = \sum_{\lambda} p(\lambda) \left(\sum_{a}  D_{A}(a|x \lambda) \right) \left( \sum_{b} b D_{B}(b|\langle a_{x}^{\lambda} \rangle y \lambda) \right) \nonumber\\
    & = \sum_{\lambda} p(\lambda) \langle b_{y, \langle a_{x}^{\lambda} \rangle}^{\lambda} \rangle,
    \end{align}
    where the nonsignalling property was used.
    Altogether, \cref{eq:CorrelatorsOutLHV-alice,eq:CorrelatorsOutLHV-bob,eq:CorrelatorsOutLHV-joint} are satisfied.
    
    To prove the "$\Leftarrow$" direction, let us first define:
    \begin{subequations}
    \begin{align}
        &p_{A}(a|x \lambda) = \frac{1 + a\langle a_{x}^{\lambda} \rangle}{2}, \\
        &p_{B}(b|\pm1 y \lambda) = \frac{1 + b \langle b_{y,\pm1}^{\lambda} \rangle}{2} .
    \end{align}
    \end{subequations}
    It follows that $p_{A}(a|x \lambda)$ and $p_{B}(b|\pm1 y \lambda)$ are valid probability distributions.
    Moreover, $a\langle a_{x}^{\lambda} \rangle = 2p_{A}(a|x \lambda) - 1$, and $b\langle b_{y,\pm 1}^{\lambda} \rangle =  2p_{B}(b|\pm1 y \lambda) - 1$.
    Together with \cref{eq:CorrelatorsOutLHV}, these can be used to construct an LHV+Out model for $p(ab|xy)$ by:
    \begin{align}
        &p(ab|xy) = \frac{1}{4} \left[ 1 + a \langle a_{x}  \rangle + b \langle b_{y} \rangle + ab \langle a_{x}b_{y} \rangle \right] \nonumber\\
        & =\frac{1}{4} \sum_{\lambda} p(\lambda) \left[ 1 + a \langle a_{x}^{\lambda} \rangle + b \langle b_{y, \langle a_{x}^{\lambda} \rangle}^{\lambda} \rangle + ab \langle a_{x}^{\lambda} \rangle \langle b_{y, \langle a_{x}^{\lambda} \rangle}^{\lambda} \rangle \right] \nonumber\\
        & = \sum_{\lambda} p(\lambda) p_{A}(a|x \lambda) p_{B}(b|\langle a_{x}^{\lambda} \rangle y \lambda) \nonumber \\
        & = \sum_{\lambda} p(\lambda) p_{A}(a|x \lambda) p_{B}(b|a y \lambda).
    \end{align}
    Thus, the result follows.
\end{proof}

We now specialise the correlator representation to LHV+Out models for Werner states under projective measurements.
Specifically, let Alice and Bob share the Werner state  $W(v) = v \dyad{\psi^-} + (1-v)\tfrac{\eye}{4}$, with $\ket{\psi^-} = (\ket{01} - \ket{10})/\sqrt{2}$, and suppose that each party performs a local projective measurement. Since both parties hold two-dimensional systems, their projective measurements can be represented by sets of normalised Bloch vectors $\set{\vec{u}_x}_{x = 1}^{n_x}$ and $\set{\vec{u}_y}_{y = 1}^{n_y}$, respectively.
Each vector $\vec{u}$ is then associated with a measurement composed of the two effects $\Pi_{\vec{u}, c} = \frac{1}{2}\left(\eye + c \ \vec{u} \cdot \vec{\sigma}\right)$, where \(c \in \set{\pm1}\) and \(\vec{\sigma} = (X, Y, Z)\) denotes the Pauli matrices.
A straightforward calculation shows that, when applied to $W(v)$, these measurements yield the expectation values
\begin{subequations}
\label{eq:lhvoutcorrelators}
\begin{align}
    & \langle a_{x} \rangle \equiv \sum_{a}a \,p_{A}(a|x) = \sum_{a}a \trace\left[ W(v) \cdot (\Pi_{\vec{u}_x, a} \otimes \eye)\right] = 0, \\
    &\langle b_{y} \rangle \equiv \sum_{b}b \, p_{B}(b|y) = \sum_{b}b \trace\left[ W(v) \cdot ( \eye \otimes \Pi_{\vec{u}_y, b})\right] = 0,
    \label{eq:1BodyCorrBob} \\				
    & \langle a_{x} b_{y} \rangle \equiv \sum_{a,b} ab \,p(ab|xy) = \sum_{a,b}ab \trace\left[ W(v) \cdot (\Pi_{\vec{u}_x, a} \otimes  \Pi_{\vec{u}_y, b})\right] = -v  \vec{u}_x \cdot \vec{u}_y,
    \label{eq:2BodyCorr}
\end{align}
\end{subequations}
where the left-hand side of the above equations is given by the LHV+Out decomposition discussed in \cref{prop:LHV+Out-correlators}.

Our implementation of a method for constructing LHV+Out models finds this decomposition in terms of explicit deterministic strategies $\lambda$ and their weights $p(\lambda)$.
Notice that since the model is explicitly constructed, it becomes independent of the method itself and can be verified straightforwardly.
A \textsc{Julia} script to perform this verification, together with the measurements and deterministic strategies, is provided in an online repository \cite{zenodo_sm} and shown below for completeness.

\subsection{Verification of the LHV+Out model}
\label{app:numerics-verification}

\noindent The code reproduced below can be used in \textsc{Julia}, together with the files
\texttt{werner\_nomarg\_sqrt2.dat} and \texttt{HQV20\_v.dat} \cite{zenodo_sm} containing the deterministic strategies and the measurement vectors, respectively, to verify the LHV+Out model for a nonlocal two qubit Werner state
under all projective measurements on a hemisphere of Alice's Bloch ball and all projective measurements for Bob (see \cref{app:proof-result-3}).

    \hypertarget{constructing-the-quantum-behaviour}{%
\subsubsection{Constructing the quantum
behaviour}\label{constructing-the-quantum-behaviour}}

    We choose the same set of \(n_m\) measurements for both Alice and
Bob, represented by the unit vectors \(\vec{u}_x\) where
\(x \in \{1, \ldots, n_m\}\) and correspondingly by \(\vec{u}_y\) for
Bob. A set with \(401\) measurements is provided in the file
\texttt{HQV20\_v.dat}, which we open by executing:

    \begin{tcolorbox}[breakable, size=fbox, boxrule=1pt, pad at break*=1mm,colback=cellbackground, colframe=cellborder]
\prompt{In}{incolor}{2}{\boxspacing}
\begin{Verbatim}[commandchars=\\\{\}]
\PY{n}{meas} \PY{o}{=} \PY{n}{deserialize}\PY{p}{(}\PY{l+s}{\PYZdq{}}\PY{l+s}{./HQV20\PYZus{}v.dat}\PY{l+s}{\PYZdq{}}\PY{p}{)}\PY{p}{;}
\PY{n}{nx} \PY{o}{=} \PY{n}{size}\PY{p}{(}\PY{n}{meas}\PY{p}{,} \PY{l+m+mi}{1}\PY{p}{)}\PY{p}{;} \PY{n}{ny} \PY{o}{=} \PY{n}{nx}\PY{p}{;}
\end{Verbatim}
\end{tcolorbox}

    The measurement's Bloch vectors are represented as a matrix with one vector per row:

    \begin{tcolorbox}[breakable, size=fbox, boxrule=1pt, pad at break*=1mm,colback=cellbackground, colframe=cellborder]
\prompt{In}{incolor}{3}{\boxspacing}
\begin{Verbatim}[commandchars=\\\{\}]
\PY{n}{meas}
\end{Verbatim}
\end{tcolorbox}

    \begin{Verbatim}[commandchars=\\\{\}]
401×3 Matrix\{Float64\}:
 1.0        0.0       0.0
 0.987688   0.156434  0.0
 0.951057   0.309017  0.0
 0.891007   0.45399   0.0
 0.809017   0.587785  0.0
 0.707107   0.707107  0.0
 0.587785   0.809017  0.0
 0.45399    0.891007  0.0
 0.309017   0.951057  0.0
 0.156434   0.987688  0.0
 \vdots                    
 0.309017  -0.951057  0.0
 0.45399   -0.891007  0.0
 0.587785  -0.809017  0.0
 0.707107  -0.707107  0.0
 0.809017  -0.587785  0.0
 0.891007  -0.45399   0.0
 0.951057  -0.309017  0.0
 0.987688  -0.156434  0.0
 0.0        0.0       1.0
    \end{Verbatim}

    To be valid measurements, they must all be normalised, and we also want
them to be on the $z \geq 0$ hemisphere:

    \begin{tcolorbox}[breakable, size=fbox, boxrule=1pt, pad at break*=1mm,colback=cellbackground, colframe=cellborder]
\prompt{In}{incolor}{14}{\boxspacing}
\begin{Verbatim}[commandchars=\\\{\}]
\PY{n+nd}{@assert} \PY{n}{all}\PY{p}{(}\PY{n}{norm}\PY{p}{(}\PY{n}{meas}\PY{p}{[}\PY{n}{i}\PY{p}{,} \PY{o}{:}\PY{p}{]}\PY{p}{)} \PY{o}{==} \PY{l+m+mf}{1.0} \PY{k}{for} \PY{n}{i} \PY{k}{in} \PY{n}{size}\PY{p}{(}\PY{n}{meas}\PY{p}{,} \PY{l+m+mi}{1}\PY{p}{)}\PY{p}{)}
\PY{n+nd}{@assert} \PY{n}{all}\PY{p}{(}\PY{n}{meas}\PY{p}{[}\PY{n}{i}\PY{p}{,} \PY{l+m+mi}{3}\PY{p}{]} \PY{o}{\PYZgt{}=} \PY{l+m+mi}{0.0} \PY{k}{for} \PY{n}{i} \PY{k}{in} \PY{n}{size}\PY{p}{(}\PY{n}{meas}\PY{p}{,} \PY{l+m+mi}{1}\PY{p}{)}\PY{p}{)}
\end{Verbatim}
\end{tcolorbox}

    As discussed in \cref{app:correlator-notation}, on a Werner state \(W(v)\) these measurements
yield the expectation values \begin{equation*}
\langle a_x \rangle = 0, \quad \langle b_y \rangle = 0, \quad \langle a_x b_y \rangle = -v \ \vec{u}_x \cdot \vec{u}_y .
\end{equation*}

Using a visibility \(v = 0.7071\), we construct a matrix with the
correlators \(\langle a_x b_y \rangle\), representing the quantum behaviour of this state under the chosen measurements:

    \begin{tcolorbox}[breakable, size=fbox, boxrule=1pt, pad at break*=1mm,colback=cellbackground, colframe=cellborder]
\prompt{In}{incolor}{5}{\boxspacing}
\begin{Verbatim}[commandchars=\\\{\}]
\PY{n}{vis} \PY{o}{=} \PY{l+m+mf}{0.7071}\PY{p}{;}
\PY{n}{p} \PY{o}{=} \PY{o}{\PYZhy{}}\PY{n}{vis} \PY{o}{.*} \PY{n}{meas} \PY{o}{*} \PY{n}{meas}\PY{o}{\PYZsq{}}\PY{p}{;}
\end{Verbatim}
\end{tcolorbox}
This represents the right-hand side of \cref{eq:2BodyCorr}

\smallskip
\begin{remark}
Since for the Werner state the marginal expectations
vanish, it will suffice to construct an LHV+Out model that reproduces only the
correlator terms \(\langle a_x b_y \rangle\). This holds because any
model that reproduces \(\langle a_x b_y \rangle\) can be extended to a
model with vanishing marginals and the same correlators. Explicitly, one
can duplicate all deterministic strategies and:

\begin{enumerate}
\def\labelenumi{\arabic{enumi}.}
\item
  Divide the weight \(p(\lambda)\) of each strategy by two.
\item
  In the duplicated strategies, flip the signs of all outcomes
  \(\langle a_x^\lambda \rangle\),
  \(\langle b_{y, -1}^\lambda \rangle\), and
  \(\langle b_{y, +1}^\lambda \rangle\).
\end{enumerate}
\end{remark}

    \hypertarget{constructing-the-model}{%
\subsubsection{Constructing a model for the behaviour}\label{constructing-the-model}}

We will now check whether the quantum behaviour \(\boldsymbol{p}\)
defined above has an LHV+Out model. This amounts to finding a
probability distribution \(p(\lambda)\), a set of deterministic
expectations
\(\{ \langle a_{x}^{\lambda} \rangle \}_{x, \lambda} \subseteq \{-1,1\}\)
for Alice, and for Bob the sets
\(\{ \langle b_{y,1}^{\lambda} \rangle \}_{y, \lambda}, \{ \langle b_{y,-1}^{\lambda} \rangle \}_{y, \lambda} \subseteq \{-1,1\}\),
reproducing the quantum behaviour through
\(\langle a_x b_y \rangle = \sum_\lambda p(\lambda) \langle a_{x}^{\lambda} \rangle \langle b_{y, \langle a_{x}^{\lambda} \rangle}^{\lambda} \rangle\), as in the left-hand side of \cref{eq:lhvoutcorrelators}.

\smallskip
The deterministic strategies and weights of our model are stored in the file \texttt{werner\_nomarg\_sqrt2.dat} and
has four components: \texttt{as}, \texttt{bms}, \texttt{bps}, and
\texttt{weights}. These correspond, respectively, to Alice's and Bob's
deterministic assignments (Bob's assignments are split between the cases where he receives $\pm 1$ from Alice), and the probability distribution over the
hidden variables.
    Let us load the model and verify it is valid:

    \begin{tcolorbox}[breakable, size=fbox, boxrule=1pt, pad at break*=1mm,colback=cellbackground, colframe=cellborder]
\prompt{In}{incolor}{11}{\boxspacing}
\begin{Verbatim}[commandchars=\\\{\}]
\PY{n}{model} \PY{o}{=} \PY{n}{deserialize}\PY{p}{(}\PY{l+s}{\PYZdq{}}\PY{l+s}{./werner\PYZus{}nomarg\PYZus{}sqrt2.dat}\PY{l+s}{\PYZdq{}}\PY{p}{)}\PY{p}{;}
\PY{n}{weights} \PY{o}{=} \PY{n}{model}\PY{p}{[}\PY{l+s}{\PYZdq{}}\PY{l+s}{weights}\PY{l+s}{\PYZdq{}}\PY{p}{]}\PY{p}{;}
\PY{n}{as} \PY{o}{=} \PY{n}{model}\PY{p}{[}\PY{l+s}{\PYZdq{}}\PY{l+s}{as}\PY{l+s}{\PYZdq{}}\PY{p}{]}\PY{p}{;}
\PY{n}{bms} \PY{o}{=} \PY{n}{model}\PY{p}{[}\PY{l+s}{\PYZdq{}}\PY{l+s}{bms}\PY{l+s}{\PYZdq{}}\PY{p}{]}\PY{p}{;}
\PY{n}{bps} \PY{o}{=} \PY{n}{model}\PY{p}{[}\PY{l+s}{\PYZdq{}}\PY{l+s}{bps}\PY{l+s}{\PYZdq{}}\PY{p}{]}\PY{p}{;}

\PY{n+nd}{@assert} \PY{n}{all}\PY{p}{(}\PY{n}{w} \PY{o}{\PYZgt{}=} \PY{l+m+mi}{0} \PY{k}{for} \PY{n}{w} \PY{k}{in} \PY{n}{weights}\PY{p}{)} \PY{o}{\PYZam{}\PYZam{}} \PY{n}{isapprox}\PY{p}{(}\PY{n}{sum}\PY{p}{(}\PY{n}{weights}\PY{p}{)}\PY{p}{,} \PY{l+m+mf}{1.0}\PY{p}{)} \PY{c}{\PYZsh{} normalized weights?}
\PY{n+nd}{@assert} \PY{n}{all}\PY{p}{(}\PY{n}{x} \PY{o}{\PYZhy{}\PYZgt{}} \PY{n}{x} \PY{o}{==} \PY{l+m+mi}{0} \PY{o}{||} \PY{n}{x} \PY{o}{==} \PY{l+m+mi}{1}\PY{p}{,} \PY{p}{[}\PY{n}{as}\PY{p}{;} \PY{n}{bms}\PY{p}{;} \PY{n}{bps}\PY{p}{]}\PY{p}{)} \PY{c}{\PYZsh{} all outcomes in \PYZob{}0, 1\PYZcb{}?}
\PY{n+nd}{@assert} \PY{n}{size}\PY{p}{(}\PY{n}{as}\PY{p}{,} \PY{l+m+mi}{2}\PY{p}{)} \PY{o}{==} \PY{n}{nx} \PY{o}{\PYZam{}\PYZam{}} \PY{n}{size}\PY{p}{(}\PY{n}{bms}\PY{p}{,} \PY{l+m+mi}{2}\PY{p}{)} \PY{o}{==} \PY{n}{ny} \PY{c}{\PYZsh{} nof. columns equal to the nof. settings?}
\PY{n+nd}{@assert} \PY{n}{size}\PY{p}{(}\PY{n}{bps}\PY{p}{)} \PY{o}{==} \PY{n}{size}\PY{p}{(}\PY{n}{bms}\PY{p}{)} \PY{c}{\PYZsh{} equal nof. \PYZdq{}+\PYZdq{} and \PYZdq{}\PYZhy{}\PYZdq{} meas. and strategies for Bob}
\PY{n+nd}{@assert} \PY{n}{size}\PY{p}{(}\PY{n}{bps}\PY{p}{,} \PY{l+m+mi}{1}\PY{p}{)} \PY{o}{==} \PY{n}{length}\PY{p}{(}\PY{n}{weights}\PY{p}{)} \PY{o}{\PYZam{}\PYZam{}} \PY{n}{size}\PY{p}{(}\PY{n}{as}\PY{p}{,} \PY{l+m+mi}{1}\PY{p}{)} \PY{o}{==} \PY{n}{length}\PY{p}{(}\PY{n}{weights}\PY{p}{)}

\PY{c}{\PYZsh{} the outcomes are represented as 0/1; these lines convert them to \PYZhy{}1/+1:}
\PY{n}{as} \PY{o}{=} \PY{p}{(}\PY{l+m+mi}{2} \PY{o}{*} \PY{n}{as} \PY{o}{.\PYZhy{}} \PY{l+m+mi}{1}\PY{p}{)}\PY{p}{;}
\PY{n}{bms} \PY{o}{=} \PY{p}{(}\PY{l+m+mi}{2} \PY{o}{*} \PY{n}{bms} \PY{o}{.\PYZhy{}} \PY{l+m+mi}{1}\PY{p}{)}\PY{p}{;}
\PY{n}{bps} \PY{o}{=} \PY{p}{(}\PY{l+m+mi}{2} \PY{o}{*} \PY{n}{bps} \PY{o}{.\PYZhy{}} \PY{l+m+mi}{1}\PY{p}{)}\PY{p}{;}
\end{Verbatim}
\end{tcolorbox}

    From this data, we can construct a matrix $\bm{q}$ with elements
\(\sum_\lambda w_\lambda \langle a_x^\lambda \rangle \langle b_{y, \langle a_{x}^{\lambda} \rangle}^{\lambda} \rangle\), to be later compared against $\bm{p}$.
First, let us make a helper function that takes a single deterministic
strategy (that is, a single term in the sum above) represented by
\texttt{weight}, \texttt{a}, \texttt{bm} and \texttt{bp}, and adds the
corresponding correlations matrix
\(w_\lambda \langle a_x^\lambda \rangle \langle b_{y, \langle a_x \rangle}^\lambda \rangle\)
to a given matrix \texttt{q}:

    \begin{tcolorbox}[breakable, size=fbox, boxrule=1pt, pad at break*=1mm,colback=cellbackground, colframe=cellborder]
\prompt{In}{incolor}{7}{\boxspacing}
\begin{Verbatim}[commandchars=\\\{\}]
\PY{k}{function} \PY{n}{add\PYZus{}weighted\PYZus{}outcome\PYZus{}matrix!}\PY{p}{(}\PY{n}{q}\PY{p}{,} \PY{n}{weight}\PY{p}{,} \PY{n}{a}\PY{p}{,} \PY{n}{bm}\PY{p}{,} \PY{n}{bp}\PY{p}{)}
    \PY{k}{for} \PY{n}{x} \PY{k}{in} \PY{n}{axes}\PY{p}{(}\PY{n}{q}\PY{p}{,} \PY{l+m+mi}{1}\PY{p}{)}
        \PY{k}{if} \PY{n}{a}\PY{p}{[}\PY{n}{x}\PY{p}{]} \PY{o}{==} \PY{l+m+mi}{1}
            \PY{k}{for} \PY{n}{y} \PY{k}{in} \PY{n}{axes}\PY{p}{(}\PY{n}{q}\PY{p}{,} \PY{l+m+mi}{2}\PY{p}{)}
                \PY{n}{q}\PY{p}{[}\PY{n}{x}\PY{p}{,} \PY{n}{y}\PY{p}{]} \PY{o}{+=} \PY{n}{weight} \PY{o}{*} \PY{n}{bp}\PY{p}{[}\PY{n}{y}\PY{p}{]} \PY{c}{\PYZsh{} When a\PYZus{}x = 1 then a\PYZus{}x * b\PYZus{}y = b\PYZus{}\PYZob{}y,1\PYZcb{}}
            \PY{k}{end}
        \PY{k}{elseif} \PY{n}{a}\PY{p}{[}\PY{n}{x}\PY{p}{]} \PY{o}{==} \PY{o}{\PYZhy{}}\PY{l+m+mi}{1}
            \PY{k}{for} \PY{n}{y} \PY{k}{in} \PY{n}{axes}\PY{p}{(}\PY{n}{q}\PY{p}{,} \PY{l+m+mi}{2}\PY{p}{)}
                \PY{n}{q}\PY{p}{[}\PY{n}{x}\PY{p}{,} \PY{n}{y}\PY{p}{]} \PY{o}{\PYZhy{}=} \PY{n}{weight} \PY{o}{*} \PY{n}{bm}\PY{p}{[}\PY{n}{y}\PY{p}{]} \PY{c}{\PYZsh{} When a\PYZus{}x = \PYZhy{}1 then a\PYZus{}x * b\PYZus{}y = \PYZhy{}b\PYZus{}\PYZob{}y,\PYZhy{}1\PYZcb{}}
            \PY{k}{end}
        \PY{k}{else}
            \PY{n+nd}{@error} \PY{l+s}{\PYZdq{}}\PY{l+s}{Invalid element in a matrix.}\PY{l+s}{\PYZdq{}}
        \PY{k}{end}
    \PY{k}{end}
    \PY{k}{return} \PY{n}{q}
\PY{k}{end}{;}
\end{Verbatim}
\end{tcolorbox}

    Starting from \(\boldsymbol{q} = 0\), we run the function above for each
deterministic strategy:

    \begin{tcolorbox}[breakable, size=fbox, boxrule=1pt, pad at break*=1mm,colback=cellbackground, colframe=cellborder]
\prompt{In}{incolor}{8}{\boxspacing}
\begin{Verbatim}[commandchars=\\\{\}]
\PY{n}{q} \PY{o}{=} \PY{n}{zeros}\PY{p}{(}\PY{n}{eltype}\PY{p}{(}\PY{n}{weights}\PY{p}{)}\PY{p}{,} \PY{n}{size}\PY{p}{(}\PY{n}{p}\PY{p}{)}\PY{p}{)}\PY{p}{;}
\PY{k}{for} \PY{n}{i} \PY{k}{in} \PY{l+m+mi}{1}\PY{o}{:}\PY{n}{length}\PY{p}{(}\PY{n}{weights}\PY{p}{)} \PY{c}{\PYZsh{} for each deterministic strategy add its matrix to q:}
    \PY{n+nd}{@views} \PY{n}{add\PYZus{}weighted\PYZus{}outcome\PYZus{}matrix!}\PY{p}{(}\PY{n}{q}\PY{p}{,} \PY{n}{weights}\PY{p}{[}\PY{n}{i}\PY{p}{]}\PY{p}{,} \PY{n}{as}\PY{p}{[}\PY{n}{i}\PY{p}{,} \PY{o}{:}\PY{p}{]}\PY{p}{,} \PY{n}{bms}\PY{p}{[}\PY{n}{i}\PY{p}{,} \PY{o}{:}\PY{p}{]}\PY{p}{,} \PY{n}{bps}\PY{p}{[}\PY{n}{i}\PY{p}{,} \PY{o}{:}\PY{p}{]}\PY{p}{)}
\PY{k}{end}\PY{p}{;}
\end{Verbatim}
\end{tcolorbox}

\hypertarget{exact-decomposition}{%
\subsubsection{Exact decomposition}\label{exact-decomposition}}

By construction, \(\boldsymbol{q}\) is an LHV+Out behaviour. How well
does it reproduce the quantum behaviour in \(\boldsymbol{p}\)? To
measure that, we define
\(\epsilon = \lVert \boldsymbol{p} - \boldsymbol{q} \rVert_2\):

    \begin{tcolorbox}[breakable, size=fbox, boxrule=1pt, pad at break*=1mm,colback=cellbackground, colframe=cellborder]
\prompt{In}{incolor}{9}{\boxspacing}
\begin{Verbatim}[commandchars=\\\{\}]
\PY{n}{epsilon} \PY{o}{=} \PY{n}{norm}\PY{p}{(}\PY{n}{p} \PY{o}{\PYZhy{}} \PY{n}{q}\PY{p}{)} \PY{c}{\PYZsh{} this is the Euclidean norm}
\PY{n}{print}\PY{p}{(}\PY{n}{epsilon}\PY{p}{)}
\end{Verbatim}
\end{tcolorbox}

    \begin{Verbatim}[commandchars=\\\{\}]
0.00019999656135527604
    \end{Verbatim}

    As we can see, \(\boldsymbol{q}\) closely reproduces \(\boldsymbol{p}\),
but not yet exactly. To get an exact model for a quantum behaviour, consider
instead the behaviour given by \[
    \nu \boldsymbol{p} = \nu \boldsymbol{q} + (1 - \nu) \boldsymbol{y} ,
\] where \(\nu\) is a scalar and \(\boldsymbol{y}\) is any LHV+Out
behaviour. If we find any \(\nu\) and \(\boldsymbol{y}\) satisfying this
equation, we can guarantee that \(\nu \boldsymbol{p}\) has an exact
LHV+Out representation.
Lemma 1 in Ref. \cite{DIB+23} shows that any correlations matrix
\(\boldsymbol{y}\) with \(\lVert \boldsymbol{y} \rVert_2 \leq 1\) is an
LHV behaviour. As
\(\mathrm{LHV}(n_x, n_y, n_a, n_b) \subset \text{LHV+Out}(n_x, n_y, n_a, n_b)\),
the same \(\boldsymbol{y}\) is also LHV+Out. This condition on
\(\boldsymbol{y}\) can be satisfied by choosing
\(\nu = 1 / (1 + \epsilon)\), which leads to an exact model for the
behaviour of the state \(W(\nu v)\), under the same set of measurements.

    \hypertarget{converting-into-a-model-for-the-state}{%
\subsubsection{Converting into a model for the
state}\label{converting-into-a-model-for-the-state}}

    Lastly, we must argue that this model is valid not only for the
behaviours of \(W(\nu v)\) under our particular choice of measurements,
but rather, for some state \(W(\tilde{\nu} v)\) under \emph{all}
possible projective measurements by Bob, and all projective measurements by Alice on a hemisphere of the Bloch ball. As explained in \cref{app:proof-result-3}, one
can consider the convex hull of the measurement vectors, compute the
radius \(\eta\) of the largest hemisphere that can be inscribed in it,
then observe that the state \(W(\eta^2 \nu v)\) admits the model for all
projective measurements on the hemisphere of Alice's and the entire sphere for Bob. For our particular
set of measurements, \(\eta\) is known to be \(\cos(\pi / 40)^2\)
\cite{hirsch2017betterlocalhidden}.
Altogether, this leads to the following visibility for which the Werner
state has an LHV+Out model for all measurements in the hemisphere:

    \begin{tcolorbox}[breakable, size=fbox, boxrule=1pt, pad at break*=1mm,colback=cellbackground, colframe=cellborder]
\prompt{In}{incolor}{10}{\boxspacing}
\begin{Verbatim}[commandchars=\\\{\}]
\PY{n}{nu} \PY{o}{=} \PY{l+m+mi}{1} \PY{o}{/} \PY{p}{(}\PY{l+m+mi}{1} \PY{o}{+} \PY{n}{epsilon}\PY{p}{)}
\PY{n}{eta\PYZus{}squared} \PY{o}{=} \PY{n}{cos}\PY{p}{(}\PY{n+nb}{pi}\PY{o}{/}\PY{l+m+mi}{40}\PY{p}{)}\PY{o}{\PYZca{}}\PY{l+m+mi}{4}
\PY{n}{final\PYZus{}visibility} \PY{o}{=} \PY{n}{nu} \PY{o}{*} \PY{n}{eta\PYZus{}squared} \PY{o}{*} \PY{n}{vis}
\end{Verbatim}
\end{tcolorbox}

    \begin{Verbatim}[commandchars=\\\{\}]
0.6982815667392431
    \end{Verbatim}

    Because there exists a Bell inequality which is violated by the Werner
state with visibility \(0.69604\) \cite{designolle2024better}, this LHV+Out state is
nonlocal, thus proving there exist nonlocal states which are LHV+Out for
every projective measurement on a hemisphere.

\end{document}